\documentclass{article}
\usepackage{amsmath, amssymb}\usepackage{extarrows}
\usepackage{amsfonts}
\usepackage{appendix}
\usepackage{graphicx}
\usepackage{xcolor}
\usepackage{url}
\usepackage{bm}
\usepackage[all]{xy}
\usepackage{arydshln}
\usepackage{colortbl}
\usepackage{tabularx}
\usepackage{makecell}
\usepackage{textcomp}

\usepackage{multirow}
\usepackage{diagbox}
\usepackage{threeparttable}

\usepackage{cases}

\newtheorem{theorem}{Theorem}
\newtheorem{lemma}{Lemma}

\newtheorem{property}{Property}

\newtheorem{proposition}{Proposition}

\newenvironment{proof}{{\noindent\it Proof}\quad}{\par}

\newcommand{\F}{\ensuremath{\mathbb F}}
\newcommand{\Z}{\ensuremath{\mathbb Z}}

\newcommand{\C}{\ensuremath{\mathbb C}}

\newcommand{\ls}[1]
{\dimen0=\fontdimen6\the\font\lineskip=#1\dimen0
	\advance\lineskip.5\fontdimen5\the\font
	\advance\lineskip-\dimen0
	\lineskiplimit=0.9\lineskip
	\baselineskip=\lineskip
	\advance\baselineskip\dimen0
	\normallineskip\lineskip\normallineskiplimit\lineskiplimit
	\normalbaselineskip\baselineskip
	\ignorespaces}

\begin{document}
	
	\bibliographystyle{abbrv}
	
	\title{The $q$-ary Golay complementary arrays
		of size $\bm{2}^{(m)}$ are standard
	}

\author{Erzhong Xue, Zilong Wang\\
	\small  State Key Laboratory of Integrated Service Networks\\[-0.8ex]
	\small School of Cyber Engineering, Xidian University\\
	\small Xi'an, 710071, China\\
	\small\tt 2524384374@qq.com, zlwang@xidian.edu.cn\\
}
	

	\maketitle

	\thispagestyle{plain} \setcounter{page}{1}

	\ls{1.5}

\begin{abstract}
To find the non-standard binary Golay complementary sequences (GCSs) of length $2^{m}$ or  theoretically prove  the nonexistence of them are still open.
Since it has been shown 
that all the standard $q$-ary (where $q$ is even) GCSs of length $2^m$ can be obtained by standard $q$-ary Golay complementary array pair (GAP) of dimension $m$ and size $2\times 2 \times \cdots \times 2$ (abbreviated to size  $\bm{2}^{(m)}$),
it's natural to ask whether all the $q$-ary GAP of size $\bm{2}^{(m)}$ are standard.
We give a positive answer to this question.

\end{abstract}

\section{Introduction} \label{section:intro}

The binary GCSs were first introduced by Golay \cite{Golay1951Static} in the context of infrared spectrometry, and were generalized to $q$-ary cases later on.
A $q$-ary sequence $\bm{f}$ of length $L$ is defined as 
\[\bm{f}=(f(0), f(1),\cdots, f(L-1)),\]
where each entry $f(t)$ belongs to $\mathbb{Z}_q$ for $0\leq{t}<L$.
The {\em aperiodic auto-correlation} of sequence $\bm{f}$ at shift $\tau$ $(-L< \tau < L)$ is defined by
\[C_{{f}}(\tau):=\sum_{t} \zeta^{{f}(t+\tau)-{f}(t)},
\]
where $\zeta=e^{2\pi\sqrt{-1}/q}$ and $\zeta^{{f}(t+\tau)-{f}(t)} := 0$ if ${f}(t+\tau)$ or ${f}(t)$ is not defined.
%
A pair of sequences $\{\bm{f_1},\bm{f_2}\}$ is called a {\em Golay complementary pair (GCP)} \cite{Golay1951Static} if
\[
C_{{f_1}}(\tau)+C_{{f_2}}(\tau)=0,\;\text{for all}\;\tau \neq 0.
\]
Each sequence in such a pair is called a {\em GCS}.

Repeated application of Turyn's construction \cite{Turyn1974}, beginning with binary GCPs of lengths $2$, $10$ and $26$ given in literature, can be used to construct GCPs for all lengths $2^{a}10^{b}26^{c}$, $a,b,c\geq0$, which are all the known lengths of binary GCPs up to now. In 1999, Davis and Jedwab \cite{Davis1999Peak} gave a explicit expression of $(q=2^h)$-ary GCPs of length $2^{m}$ based on Generalized Boolean functions. Paterson \cite{Paterson00} showed the same construction holds without modification for any even $q$, which is given bellow.

For any permutation $\pi$ of $\{1,2,\dots,m\}$, and any choice of constants $c'\in\mathbb{Z}_{q}$, $c_{k}\in\mathbb{Z}_{q}$ $(1\leq k\leq m)$, define the pair of generalized Boolean functions
\begin{equation}\label{eq: thm_fg}
	\left\{
	\begin{aligned}
		f(x_1,x_2,\cdots,x_m)&=\frac{q}{2}\sum_{k=1}^{m-1}x_{\pi(k)}x_{\pi(k+1)}+ \sum_{k=1}^{m}c_{k}x_{k}+c_{0},\\
		g(x_1,x_2,\cdots,x_m)&=f(x_1,x_2,\cdots,x_m)+\frac{q}{2}x_{\pi(1)}+c'.
	\end{aligned}
	\right.
\end{equation}
The sequence pais $(f(t),g(t))$ projected from $(f(x_1,x_2,\cdots,x_m),g(x_1,x_2,\cdots,x_m))$ by restricting $t=\sum_{k=1}^{m}2^{k-1}x_{k}$ form GCP of length $2^{m}$, and have been referred to as the {\em standard} since then.
The sequences in  standard GCPs are called  {\em standard GCSs}.

Numerical evidence \cite{Borwein2004A} suggests that there are no other binary GCSs of length $2^m$ when $m$ does not exceed 100. To find the non-standard binary GCSs of length $2^m$ ($m>100$) or to theoretically prove  the nonexistence of binary GCSs are still open.

In 2005, Li and Chu \cite{Li2005More} discovered $ 1024 $ non-standard quaternary GCSs of length $16$ by computer search. Fiedler and Jedwab \cite{Fiedler2006How} explained shortly afterwards by means of \lq \lq crossover\rq \rq \ quaternary standard GCPs of length $8$. Furthermore, a framework to construct GCSs presented in \cite{Fiedler2008A} simplified the previous approaches to construct GCPs from the viewpoint of \lq \lq array\rq \rq .

Golay (complementary) arrays are firstly studied by L{\"u}ke \cite{Luke1985Welti} and Dymond \cite{Dymond1992} as a generalisation of combinatorial object, and then shown to be one most powerful method to study GCSs \cite{Jedwab2007Golay}.
An $m$-dimensional $q$-ary array $\bm{f}$ of size $2 \times 2 \times \cdots \times 2$ (or abbreviated to size  $\bm{2}^{(m)}$) can be represented by a generalized Boolean function over $\Z_{q}$
\[f(\bm{x})=f(x_1,x_2,\cdots,x_m),
\]
where for all $1\leq k\leq m$, $x_k \in \{0,1\}$.
The {\em aperiodic auto-correlation} of an array $\bm{f}$ at shift $\bm{\tau} = (\tau_1,\tau_2,\cdots,\tau_m)$, $(-1\leq\tau_i\leq 1)$, is defined by
\begin{equation*}\label{equation:eq1}
	C_{{f}}(\bm{\tau})=\sum_{\bm{x}\in \{0,1\}^m} \zeta^{f(\bm{x}+\bm{\tau})-f(\bm{x})},
\end{equation*}
where $\zeta^{f(\bm{x}+\bm{\tau})-f(\bm{x})}:=0$ if $f(\bm{x}+\bm{\tau})$ or $f(\bm{x})$ is not defined, and $\bm{x}+\bm{\tau}$ is the element-wise addition of integers.
A pair of arrays $(\bm{f}_1,\bm{f}_2)$ is called a {\em GAP} if
\begin{equation*}\label{equation:eq2}
	C_{{f}_1}(\bm{\tau})+C_{{f}_2}(\bm{\tau})=0\;\text{for all}\;\bm{\tau} \neq \bm{0}.
\end{equation*}
Each array in such a pair is called a {\em Golay complementary array (GCA)} \cite{Jedwab2007Golay}.

It has been shown in \cite{Fiedler2008Am} that  $(f(\bm{x}),g(\bm{x}))$ given in \eqref{eq: thm_fg} form GAPs of size $\bm{2}^{(m)}$, from which all the $q$-ary standard GCPs of length $2^m$ can be obtained.
We call the Golay array pairs shown in \eqref{eq: thm_fg} {\em standard}.
Similar to the case of GCPs,
it's natural to ask an inverse problem whether all the binary (or  $q$-ary) GAP of size $\bm{2}^{(m)}$ are standard.

To study this problem,
 the present authors \cite{Chai2021DCCWalsh} exhibited the Walsh and nega-Walsh spectrum distribution  of the binary and quaternary GCAs in 2021. The GCAs  can only be constructed from (generalized) Boolean functions satisfying spectral values given in \cite{Chai2021DCCWalsh}.
For instance, an $m$-dimensional binary GCA must be bent for even $ m $ and near-bent for odd $ m $ with respect to the Walsh spectrum, and it must be negaplateaued, nega-bent or negalandscape with respect to the nega spectrum.

In this paper, we further study this problem and give a positive answer to it, which is given by the following Theorem.
\begin{theorem}\label{thm: Type I}
The $q$-ary GAPs $(f(\bm{x}),g(\bm{x}))$ of size $\bm{2}^{(m)}$
must be standard, i.e., of form \eqref{eq: thm_fg}.
\end{theorem}

\section{Proof of Theorem \ref{thm: Type I}}\label{section:type-II}

\subsection{Preliminary}
The {\em generating function} of a $q$-ary array $f(\bm{x})$ is given by
\begin{equation}\label{eq: F(z)}
	F(\bm{z})=\sum_{\bm{x} \in \{0,1\}^m}\zeta^{f(\bm{x})}z_1^{x_1}z_2^{x_2}\cdots z_m^{x_m},
\end{equation}
where $\bm{z}=(z_1,z_2,\cdots,z_m)$.
Conversely, $f(\bm{x})$ is called the $q$-ary array of the  generating function $F(\bm{z})$.
\begin{property}\label{property1}
The generating function of aperiodic auto-correlation of $C_{{f}}(\bm{\tau})$ is given by
\begin{equation}\label{array-correlation}
	F(\bm{z})\cdot \overline{F}(\bm{z}^{-1})=\sum_{\bm{\tau}}C_{f}(\bm{\tau})\cdot\bm{z}^{\bm{\tau}},
\end{equation}
where $\bm{z}^{\bm{\tau}}$ means $z_1^{\tau_1}z_2^{\tau_2}\cdots z_{m}^{\tau_{m}}$,
$\overline{F}(\bm{z})$ is the conjugate of $F(\bm{z})$ in all the coefficients, and $\bm{z}^{-1}$ means $(z_1^{-1},z_2^{-1},\cdots, z_{m}^{-1})$.
\end{property}
A GAP can be alternatively defined from the generating functions.
\begin{property}\label{lem: GAP}
	$f_1(\bm{x})$ and $f_2(\bm{x})$ forms a GAP if and only if their generating functions $F_1(\bm{z})$ and $F_2(\bm{z})$ satisfy
	\begin{equation}\label{equation:eq7}
		F_1(\bm{z})\cdot \overline{F}_1(\bm{z}^{-1})+F_2(\bm{z})\cdot \overline{F}_2(\bm{z}^{-1})= 2^{m+1}.
	\end{equation}
\end{property}

Define  $f^{*}(\bm{x})$ to be $q$-ary array of size $\bm{2}^{(m)}$ given by
\begin{equation}
f^{*}(\bm{x})=f(\bar{x}_1,\bar{x}_2,\cdots,\bar{x}_m):\{0,1\}^{m}\to \mathbb{Z}_q,
\end{equation}
	where $\bar{x}_k=1-{x}_{k}\in \{0,1\}$, ($1\leq k\leq m$). $f^{*}(\bm{x})$ is called reverse array of $f(\bm{x})$.

For any $F(\bm{z})\in\C[\bm{z}]$ ($\bm{z}=(z_{1},z_{2},\dots,z_{m})$), multivariate polynomial ring over complex field, let  $d_{k}$ be the degree of $F(\bm{z})$ with respect to  $z_{k}$ for $1\leq{k}\leq{m}$.
Define polynomial 
\begin{equation}\label{eq: F*(z)}
	F^{*}(\bm{z})=\prod_{k=1}^{m}{z}_k^{d_{k}}\cdot{F}(\bm{z}^{-1}).
\end{equation}
For example,  $F(z_{1},z_{2})=c_{0}\cdot{z}_{1}^{2}z_{2}+c_{1}\cdot{z}_{1}z_{2}+c_{2}\cdot{z}_{1}+c_{3}\cdot{z}_{2}^{2}$, where $c_{0},c_{1},c_{2},c_{3}\in\C$, then $F^{*}(z_{1},z_{2})=c_{0}\cdot{z}_{2}+c_{1}\cdot{z}_{1}{z}_{2}+c_{2}\cdot{z}_{1}{z}_{2}^{2}+c_{3}\cdot{z}_{1}^{2}$.

\begin{property}\label{property: F*(z)}
Let $F(\bm{z})$ be the generating function of a $q$-ary array $f(\bm{x})$. Then $F^{*}(\bm{z})$ is the generating function of a $q$-ary array $f^{*}(\bm{x})$.
\end{property}


Since $\mathbb{C}[\bm{z}]$ is a unique factorization domain \cite[Sec. 9.6]{Algebra},
any $F(\bm{z})\in\mathbb{C}[\bm{z}]$ can be factored into a product of irreducible polynomials in $\mathbb{C}[\bm{z}]$, and any two polynomials in $\mathbb{C}[\bm{z}]$ has a great common divisor.

Denote indeterminates $(\bm{z}_1,\bm{z}_2)$ as {\em a partition of} the indeterminates $\bm{z}$, where $\bm{z}_1\cup \bm{z}_2=\bm{z}$ and $\bm{z}_1\cap\bm{z}_2=\varnothing$. 
For a  polynomial $A(\bm{z}_1,\bm{z}_2)\in \mathbb{C}[\bm{z}_1,\bm{z}_2]$, if the degree of $A(\bm{z}_1,\bm{z}_2)$ with respect to any indeterminate in $\bm{z}_1$ is larger than 0, and the degree of $A(\bm{z}_1,\bm{z}_2)$ with respect to any indeterminate in $\bm{z}_2$ is  0, we say that $A(\bm{z}_1,\bm{z}_2)$ can be  {\em succinctly} denoted by $A(\bm{z}_1)$ in this paper. 
For example, if $A(z_0, z_1, z_2, z_3)=z_1z_3+z_1+1$, $A(z_0, z_1, z_2, z_3)$ can be succinctly denoted by $A(z_1, z_3)$.
Define $\text{dim}\;\bm{z}$ as the number of indeterminates in $\bm{z}$.

Let ${F}(\bm{z})$ be the generating function of a $q$-ary array ${f}(\bm{x})$ of size $\bm{2}^{(m)}$ ($m\geq2$).
Suppose that ${F}(\bm{z})$ has a nontrivial factorization in $\C[\bm{z}]$, say, ${F}(\bm{z})={A}'(\bm{z})\cdot{B}'(\bm{z})$, and ${A}'(\bm{z})$ and ${B}'(\bm{z})$ can be  succinctly denoted by ${A}'(\bm{z}_{1})$ and ${B}'(\bm{z}_{2})$, i.e., 
\begin{equation}
	{F}(\bm{z})={A}'(\bm{z}_{1})\cdot{B}'(\bm{z}_{2}).
\end{equation}
Since the product of constant terms of ${A}'(\bm{z})$ and ${B}'(\bm{z})$ is the constant term of ${F}(\bm{z})$, which belongs to $\{\zeta^c|c\in\Z_{q}\}$, 
we can always find $\alpha\in\C$ which satisfies that the constant terms of both $\alpha\cdot{A}'(\bm{z}_{2})$ and $\alpha^{-1}\cdot{B}'(\bm{z}_{1})$ belong to $\{\zeta^c|c\in\Z_{q}\}$.
Let ${A}(\bm{z}_{1})=\alpha\cdot{A}'(\bm{z}_{1})$ and ${B}(\bm{z}_{2})=\alpha^{-1}\cdot{B}'(\bm{z}_{2})$. We call
\begin{equation}\label{eq: F=AB}
	{F}(\bm{z})={A}(\bm{z}_{1})\cdot{B}(\bm{z}_{2}),
\end{equation}
a normalized factorization.

\begin{lemma}\label{lem: F=AB}
	Let 
 ${F}(\bm{z})$ be the generating function of a $q$-ary array ${f}(\bm{x})$ of size $\bm{2}^{(m)}$ ($m\geq2$). Suppose that  ${F}(\bm{z})$ has a nontrivial normalized factorization in $\C[\bm{z}]$: 
 ${F}(\bm{z})={A}(\bm{z}_{1})\cdot{B}(\bm{z}_{2})$, where $\text{dim}\;\bm{z}_{1}=m_{1}$ and $\text{dim}\;\bm{z}_{2}=m_{2}$.
We have:

1) $(\bm{z}_1,\bm{z}_2) $ is a partition of the indeterminates $\bm{z}$,
and $m=m_1+m_2$.

2) ${A}(\bm{z}_{1})$ and ${B}(\bm{z}_{2})$ are generating functions of $q$-ary arrays ${a}(\bm{x}_{1})$ and ${b}(\bm{x}_{2})$, which are of size $\bm{2}^{(m_1)}$ and $\bm{2}^{(m_2)}$ respectively;
	
3) The functions ${f}(\bm{x})$, ${a}(\bm{x}_{1})$ and ${b}(\bm{x}_{2})$ satisfy
	\begin{equation}
		{f}(\bm{x})={a}(\bm{x}_{1})+{b}(\bm{x}_{2}).\\
	\end{equation}
\end{lemma}
\begin{proof}
1) It's obviously $\bm{z}=(\bm{z}_{1}\cup\bm{z}_{2})$.
Suppose that there exists an indeterminate $z\in(\bm{z}_{1}\cap\bm{z}_{2})$, i.e.,
the degrees  of ${A}(\bm{z}_{1})$ and ${B}(\bm{z}_{2})$ with respect to $z$ are both at least $1$.
Then the degree of ${F}(\bm{z})={A}(\bm{z}_{1})\cdot{B}(\bm{z}_{2})$ with respect to $z$ is at least $2$, which contradicts with the fact that
the degree of ${F}(\bm{z})$ with respect to $z$ is $1$.
Thus $\bm{z}_{1}\cap\bm{z}_{2}=\varnothing$.
Therefor, $\bm{z}_1$ and $\bm{z}_2$ is a partition of $\bm{z}$.
It is natural that the number of indeterminates satisfy $m=m_1+m_2$.

2) Let the polynomial expansions of  $F(\bm{z})$, $A(\bm{z}_{1})$ and $B(\bm{z}_{2})$ be given by
\begin{equation}
	\left\{
	\begin{aligned}
	F(\bm{z})&=\sum_{\bm{x} \in \{0,1\}^m}\zeta^{f(\bm{x})}\cdot\bm{z}^{\bm{x}},\\
	A(\bm{z}_{1})&=\sum_{\bm{x}_{1}\in \{0,1\}^{m_{1}}}{\mathcal{A}(\bm{x}_{1})}\cdot\bm{z}_{1}^{\bm{x}_{1}},\\
	{B}(\bm{z}_{2})&=\sum_{\bm{x}_{2}\in \{0,1\}^{m_2}}{\mathcal{B}(\bm{x}_{2})}\cdot\bm{z}_{2}^{\bm{x}_{2}},
	\end{aligned}
	\right.
\end{equation}
where $\mathcal{A}(\bm{x}_{1}),\mathcal{B}(\bm{x}_{2})\in\C$. ($ \bm{x}_{1}\in \{0,1\}^{m_{1}}$, $ \bm{x}_{2}\in \{0,1\}^{m_{2}}$, $\bm{x}_{1}$ and $\bm{x}_{2}$ are a partition of $\bm{x}$.)
Comparing the coefficients of each term in both sides of ${F}(\bm{z})={A}(\bm{z}_{1})\cdot{B}(\bm{z}_{2})$, 
we have 
\begin{equation}
\zeta^{f(\bm{x})}=\mathcal{A}(\bm{x}_{1})\cdot\mathcal{B}(\bm{x}_{2}).
\end{equation}
Denote the all ``$0$'' vectors by $\bm{0}\in\{0,1\}^{m}$, $\bm{0}_{1}\in\{0,1\}^{m_1}$ and $\bm{0}_{2}\in\{0,1\}^{m_2}$.
Since $\mathcal{A}(\bm{0}_{1})$ and $\mathcal{B}(\bm{0}_{2})$ are constant terms in $A(\bm{z}_{1})$ and ${B}(\bm{z}_{2})$, which belong to $\{\zeta^c|c\in\Z_{q}\}$,
let  $\mathcal{A}(\bm{0}_{1})=\zeta^{a(\bm{0}_{1})}$ and $\mathcal{B}(\bm{0}_{2})=\zeta^{b(\bm{0}_{2})}$, where $a(\bm{0}_{1}),b(\bm{0}_{2})\in\Z_{q}$.  
Notice that $\zeta^{f(\bm{0})}=\mathcal{A}(\bm{0}_{1})\cdot\mathcal{B}(\bm{0}_{2})$, we have  $a(\bm{0}_{1})+b(\bm{0}_{2})=f(\bm{0})$.
Thus $\mathcal{A}(\bm{x}_{1})$ 
	can be determined by $\zeta^{{f}(\bm{x}_{1},\bm{0}_{2})}=\mathcal{A}(\bm{x}_{1})\cdot\mathcal{B}(\bm{0}_{2})=\mathcal{A}(\bm{x}_{1})\cdot\zeta^{b(\bm{0}_{2})}$, i.e.,
\begin{equation}\label{eq: a(x)}
		\mathcal{A}(\bm{x}_{1})=\zeta^{a(\bm{x}_{1})},
	\quad\text{where}
	\quad
	a(\bm{x}_{1})={f}(\bm{x}_{1},\bm{0}_{2})-{b}(\bm{0}_{2})\in\Z_{q}.
\end{equation}
	Similarly,
\begin{equation}\label{eq: b(x)}
\mathcal{B}(\bm{x}_{2})=\zeta^{b(\bm{x}_{2})},
\quad\text{where}
\quad
b(\bm{x}_{2})={f}(\bm{0}_{1},\bm{x}_{2})-{a}(\bm{0}_{1})\in\Z_{q}.
\end{equation}
	So that ${A}(\bm{z}_{1})$ and ${B}(\bm{z}_{2})$ are generating functions of $q$-ary arrays ${a}(\bm{x}_{1})$ and ${b}(\bm{x}_{2})$ (which is given by \eqref{eq: a(x)} and \eqref{eq: b(x)}) respectively.

3) Since
${F}(\bm{z})=\sum_{\bm{x}}\zeta^{{f}(\bm{x})}\cdot\bm{z}^{\bm{x}}$, ${A}(\bm{z}_{1})=\sum_{\bm{x}_{1}}\zeta^{{a}(\bm{x}_{1})}\cdot\bm{z}_{1}^{\bm{x}_{1}}$ and ${B}(\bm{z}_{2})=\sum_{\bm{x}_{2}}\zeta^{{b}(\bm{x}_{2})}\cdot\bm{z}_{2}^{\bm{x}_{2}}$,
we have
\begin{equation}
\sum_{\bm{x}}\zeta^{{f}(\bm{x})}\cdot\bm{z}^{\bm{x}}
=\sum_{\bm{x}_{1}}\sum_{\bm{x}_{2}}\zeta^{{a}(\bm{x}_{1})+{b}(\bm{x}_{2})}\cdot\bm{z}_{1}^{\bm{x}_{1}}\cdot\bm{z}_{2}^{\bm{x}_{2}}.
\end{equation}
Comparing the coefficients of $\bm{z}^{\bm{x}}
=\bm{z}_{1}^{\bm{x}_{1}}\cdot\bm{z}_{2}^{\bm{x}_{2}}$ in both sides,
we obtain ${f}(\bm{x})={a}(\bm{x}_{1})+{b}(\bm{x}_{2})$.
\hfill\ensuremath{\square}
\end{proof}

\begin{lemma}\label{lem: (A,B)=C}
For $F(\bm{z})\in\C[\bm{z}]$,
if  $F(\bm{z})$ has factorization ${F}(\bm{z})={A}(\bm{z})\cdot{B}(\bm{z})$, then we have ${F}^{*}(\bm{z})={A}^{*}(\bm{z})\cdot{B}^{*}(\bm{z})$.
\end{lemma}
\begin{proof}
Let the degrees of $F(\bm{z})$, $A(\bm{z})$ and $B(\bm{z})$ with respect to  $z_{k}$  are 
 $d_{k}$, $a_{k}$ and $b_{k}$ respectively.
It's obviously that $a_{k}+b_{k}=d_{k}$.
Substituting  ${z}_{k}$ by ${z}_{k}^{-1}$ in ${F}(\bm{z})={A}(\bm{z})\cdot{B}(\bm{z})$, and multiply both sides by $ \prod_{k=1}^{m}{z}_k^{d_{k}}$, we obtain
\begin{equation}
\prod_{k=1}^{m}{z}_k^{d_{k}}\cdot{F}(\bm{z}^{-1})
=\prod_{k=1}^{m}{z}_k^{a_{k}}\cdot{A}(\bm{z}^{-1})\cdot
\prod_{k=1}^{m}{z}_k^{b_{k}}\cdot{B}(\bm{z}^{-1}),
\end{equation}
which complete the proof.
\hfill\ensuremath{\square}
\end{proof}

\subsection{The Proof}\label{subsec: proof}
We shall give the proof of Theorem \ref{thm: Type I} by applying mathematical induction.

\noindent
Basic Step:  It is know that Theorem \ref{thm: Type I} holds for arrays of size  $\bm{2}^{(1)}$.

\noindent
Inductive Step: Now we assume Theorem \ref{thm: Type I} holds for arrays of size  $\bm{2}^{(n)}$, where $n\leq{m}$.
From the assumption we shall deduce Theorem \ref{thm: Type I} holds for arrays of size  $\bm{2}^{(m+1)}$.

Suppose that $f(\bm{x},{x}_{m+1})$ and $g(\bm{x},{x}_{m+1})$ form a GAP of array size  $\bm{2}^{(m+1)}$, where $\bm{x}=(x_1,x_2,\cdots,x_m)\in\F_{2}^{m}$. Let ${f}_{0}(\bm{x})={f}(\bm{x},0)$, ${f}_{1}(\bm{x})={f}(\bm{x},1)$, ${g}_{0}(\bm{x})={g}(\bm{x},0)$ and ${g}_{1}(\bm{x})={g}(\bm{x},1)$. 
Then we have
\begin{numcases}{}
{f}(\bm{x},{x}_{m+1})={f}_{0}(\bm{x})(1-{x}_{m+1})+{f}_{1}(\bm{x})\cdot{x}_{m+1},\label{eq: f=f0+f1}\\
{g}(\bm{x},{x}_{m+1})={g}_{0}(\bm{x})(1-{x}_{m+1})+{g}_{1}(\bm{x})\cdot{x}_{m+1}.\label{eq: g=g0+g1}
\end{numcases}

Let  $F(\bm{z},{z}_{m+1})$ and $G(\bm{z},{z}_{m+1})$ be the generating functions of $f(\bm{x},{x}_{m+1})$ and $g(\bm{x},{x}_{m+1})$ respectively, where  $\bm{z}=(z_1,z_2,\cdots,z_m)$. According to Property \ref{lem: GAP}, we have
\begin{equation}\label{eq: F^2+G^2}
	F(\bm{z},{z}_{m+1})\cdot\overline{F}(\bm{z}^{-1},{z}_{m+1}^{-1})+G(\bm{z},{z}_{m+1})\cdot\overline{G}(\bm{z}^{-1},{z}_{m+1}^{-1})=2^{m+2}.
\end{equation}

	$F(\bm{z},{z}_{m+1})$ and 	$G(\bm{z},{z}_{m+1})$ can be expressed by 
\begin{numcases}{}
F(\bm{z},{z}_{m+1})=F_{0}(\bm{z})+F_{1}(\bm{z})\cdot{z}_{m+1},\label{eq: F(z,z)}\\
G(\bm{z},{z}_{m+1})=G_{0}(\bm{z})+G_{1}(\bm{z})\cdot{z}_{m+1},\label{eq: G(z,z)}
\end{numcases}
where $F_{0}(\bm{z})$, $F_{1}(\bm{z})$, $G_{0}(\bm{z})$ and $G_{1}(\bm{z})$ are  generating functions of $f_{0}(\bm{x})$, $f_{1}(\bm{x})$, $g_{0}(\bm{x})$ and $g_{1}(\bm{x})$ respectively.

${F}(\bm{z},{z}_{m+1})$ given in \eqref{eq: F(z,z)} times its conjugate leads to
\begin{equation}\label{eq: F^2}
	\begin{split}
		{F}(\bm{z},{z}_{m+1})\cdot\overline{F}(\bm{z}^{-1},{z}_{m+1}^{-1})&={F}_{0}(\bm{z})\cdot\overline{F}_{0}(\bm{z}^{-1})+{F}_{1}(\bm{z})\cdot\overline{F}_{1}(\bm{z}^{-1})\\
		&+{F}_{0}(\bm{z})\cdot\overline{F}_{1}(\bm{z}^{-1})\cdot{z}_{m+1}^{-1}
		+\overline{F}_{0}(\bm{z}^{-1})\cdot{F}_{1}(\bm{z})\cdot{z}_{m+1}.
	\end{split}
\end{equation}
${G}(\bm{z},{z}_{m+1})$ given in \eqref{eq: G(z,z)} times its conjugate leads to
\begin{equation}\label{eq: G^2}
	\begin{split}
		{G}(\bm{z},{z}_{m+1})\cdot\overline{G}(\bm{z}^{-1},{z}_{m+1}^{-1})&={G}_{0}(\bm{z})\cdot\overline{G}_{0}(\bm{z}^{-1})+{G}_{1}(\bm{z})\cdot\overline{G}_{1}(\bm{z}^{-1})\\
		&+{G}_{0}(\bm{z})\cdot\overline{G}_{1}(\bm{z}^{-1})\cdot{z}_{m+1}^{-1}
		+\overline{G}_{0}(\bm{z}^{-1})\cdot{G}_{1}(\bm{z})\cdot{z}_{m+1}.
	\end{split}
\end{equation}
Compare the coefficients of $1$, ${z}_{m+1}^{-1}$ and ${z}_{m+1}$ respectively between \eqref{eq: F^2}+\eqref{eq: G^2} and \eqref{eq: F^2+G^2},  we obtain
\begin{numcases}{}
{F}_{0}(\bm{z})\cdot\overline{F}_{0}(\bm{z}^{-1})+{F}_{1}(\bm{z})\cdot\overline{F}_{1}(\bm{z}^{-1})
+{G}_{0}(\bm{z})\cdot\overline{G}_{0}(\bm{z}^{-1})+{G}_{1}(\bm{z})\cdot\overline{G}_{1}(\bm{z}^{-1})=2^{m+1},\label{eq: A0+A1+B0+B1}\\
{F}_{0}(\bm{z})\cdot\overline{F}_{1}(\bm{z}^{-1})+{G}_{0}(\bm{z})\cdot\overline{G}_{1}(\bm{z}^{-1})=0,\label{eq: AA+BB=0}\\
\overline{F}_{0}(\bm{z}^{-1})\cdot{F}_{1}(\bm{z})+\overline{G}_{0}(\bm{z}^{-1})\cdot{G}_{1}(\bm{z})=0.\label{eq: FF+GG=0}
\end{numcases}

Let ${C}(\bm{z})=\left({F}_{0}(\bm{z}),{G}_{0}(\bm{z})\right)$  be the great common divisor of ${F}_{0}(\bm{z})$ and ${G}_{0}(\bm{z})$, and ${C}(\bm{z})$ be succinctly denoted by ${C}(\bm{z}_{2})$. According to Lemma \ref{lem: F=AB}, we have
\begin{numcases}{}
{F}_{0}(\bm{z})={A}(\bm{z}_{1})\cdot{C}(\bm{z}_{2}),\label{factor-1}\\
{G}_{0}(\bm{z})={B}(\bm{z}_{1})\cdot{C}(\bm{z}_{2}), \label{factor-2}
\end{numcases}
  where  $\bm{z}_1$ and $\bm{z}_2$ is a partition of the indeterminates $\bm{z}$, and ${A}(\bm{z}_{1})$, ${B}(\bm{z}_{1})$, and ${C}(\bm{z}_{2})$ are generating functions of $q$-ary array ${a}(\bm{x}_{1})$, ${b}(\bm{x}_{1})$, ${c}(\bm{x}_{2})$ respectively. 
 Without loss of generality, suppose that $\bm{z}_{1}=(z_{1},z_{2},\cdots,z_{m_1})$ $\bm{z}_{2}=(z_{m_1+1},z_{m_1+2},\cdots,z_m)$, $\bm{x}_{1}=(x_{1},x_{2},\cdots,x_{m_1})$ and $\bm{x}_{2}=(x_{m_1+1},x_{m_1+2},\cdots,x_m)$. Note that if $m_{1}=0$, we have $\bm{z}_{1}=\varnothing$,  ${A}(\bm{z}_{1})$ and ${B}(\bm{z}_{1})$ are constants in $\{\zeta^c|c\in\Z_{q}\}$, and ${a}(\bm{x}_{1})$ and ${b}(\bm{x}_{1})$ are constants over $\Z_{q}$; if $m_{1}=m$, then $\bm{z}_{2}=\varnothing$,  ${C}(\bm{z}_{2})$ is a constant in $\{\zeta^c|c\in\Z_{q}\}$, and ${c}(\bm{x}_{2})$ is a constant over $\Z_{q}$.

According to \eqref{eq: FF+GG=0}, \eqref{factor-1}, \eqref{factor-2} and  \eqref{eq: F*(z)}, we have
\begin{equation}
	\frac{{G}_{1}(\bm{z})}{{F}_{1}(\bm{z})}
	=-\frac{\overline{F}_{0}(\bm{z}^{-1})}{\overline{G}_{0}(\bm{z}^{-1})}
	=-\frac{\overline{A}(\bm{z}_{1}^{-1})}{\overline{B}(\bm{z}_{1}^{-1})}
	=-\frac{\overline{A}^{*}(\bm{z}_{1})}{\overline{B}^{*}(\bm{z}_{1})}.
\end{equation}
Since $\left({A}(\bm{z}_{1}),{B}(\bm{z}_{1})\right)=1$, we have 
 $(\overline{A}^{*}(\bm{z}_{1}),\overline{B}^{*}(\bm{z}_{1}))=1$ by Lemma \ref{lem: (A,B)=C}.
Thus there exists
\begin{equation}\label{eq: D(z)}
{D}(\bm{z}_{2})=\frac{{F}_{1}(\bm{z})}{\overline{B}^{*}(\bm{z}_{1})}
=-\frac{{G}_{1}(\bm{z})}{\overline{A}^{*}(\bm{z}_{1})}
=\left({F}_{1}(\bm{z}),{G}_{1}(\bm{z})\right),
\end{equation}
which is the generating function of a $q$-ary array ${d}(\bm{x}_{2})$.

\begin{proposition}\label{proposition1}
$({a}(\bm{x}_{1}),{b}(\bm{x}_{1}))$ and $({c}(\bm{x}_{2}),{d}(\bm{x}_{2}))$  given above are both Golay array pairs.
\end{proposition}
\begin{proof}
According to \eqref{eq: D(z)}, we have ${G}_{1}(\bm{z})=-\overline{A}^{*}(\bm{z}_{1})\cdot{D}(\bm{z}_{2})$.
According to \eqref{eq: F*(z)}, we have 
$\overline{A}^{*}(\bm{z}_{1})=\prod_{k=1}^{m_1}{z}_k\cdot\overline{A}(\bm{z}_{1}^{-1})$ and ${A}^{*}(\bm{z}_{1}^{-1})=\prod_{k=1}^{m_1}{z}_k^{-1}\cdot{A}(\bm{z}_{1})$.
Then we obtain 
\begin{equation}\label{eq: G1G1}
	\begin{split}
		{G}_{1}(\bm{z})\cdot\overline{G}_{1}(\bm{z}^{-1})
		&=\overline{A}^{*}(\bm{z}_{1})\cdot{D}(\bm{z}_{2})\cdot{A}^{*}(\bm{z}_{1}^{-1})\cdot\overline{D}(\bm{z}_{2}^{-1})\\
		&={A}(\bm{z}_{1})\cdot\overline{A}(\bm{z}_{1}^{-1})
		\cdot{D}(\bm{z}_{2})\cdot\overline{D}(\bm{z}_{2}^{-1}).
	\end{split}
\end{equation}
Similarly, we have
\begin{equation}\label{eq: F1F1}
	\begin{split}
		{F}_{1}(\bm{z})\cdot\overline{F}_{1}(\bm{z}^{-1})
		&={B}(\bm{z}_{1})\cdot\overline{B}(\bm{z}_{1}^{-1})
		\cdot{D}(\bm{z}_{2})\cdot\overline{D}(\bm{z}_{2}^{-1}).
	\end{split}
\end{equation}
According to \eqref{factor-1} and \eqref{factor-2}, we obtain 
\begin{equation}\label{eq: F0F0}
{F}_{0}(\bm{z})\cdot\overline{F}_{0}(\bm{z}^{-1})
={A}(\bm{z}_{1})\cdot\overline{A}(\bm{z}_{1}^{-1})
\cdot{C}(\bm{z}_{2})\cdot\overline{C}(\bm{z}_{2}^{-1}),
\end{equation}
and
\begin{equation}\label{eq: G0G0}
{G}_{0}(\bm{z})\cdot\overline{G}_{0}(\bm{z}^{-1})
={B}(\bm{z}_{1})\cdot\overline{B}(\bm{z}_{1}^{-1})
\cdot{C}(\bm{z}_{2})\cdot\overline{C}(\bm{z}_{2}^{-1}).
\end{equation}
By substituting  \eqref{eq: G1G1}, \eqref{eq: F1F1}, \eqref{eq: F0F0} and \eqref{eq: G0G0} into \eqref{eq: A0+A1+B0+B1}, and factoring the  polynomial, we have
\begin{equation}\label{eq: (A+B)(C+D)=2^m}
	({A}(\bm{z}_{1})\cdot\overline{A}(\bm{z}_{1}^{-1})+{B}(\bm{z}_{1})\cdot\overline{B}(\bm{z}_{1}^{-1}))\cdot
	({C}(\bm{z}_{2})\cdot\overline{C}(\bm{z}_{2}^{-1})+{D}(\bm{z}_{2})\cdot\overline{D}(\bm{z}_{2}^{-1}))=2^{m+1}.
\end{equation}
On the other hand, according to Property \ref{property1},
\begin{equation}
		{A}(\bm{z}_{1})\cdot\overline{A}(\bm{z}_{1}^{-1})
		+{B}(\bm{z}_{1})\cdot\overline{B}(\bm{z}_{1}^{-1})
		=2^{m_{1}+1}+\sum_{\bm{\tau}_{1}\ne\bm{0}}
		(C_{a}(\bm{\tau})+C_{b}(\bm{\tau}))
		\bm{z}_1^{\bm{\tau}_1},
\end{equation}
\begin{equation}
		{C}(\bm{z}_{2})\cdot\overline{C}(\bm{z}_{2}^{-1})
		+{D}(\bm{z}_{2})\cdot\overline{D}(\bm{z}_{2}^{-1})
		=2^{m-m_{1}+1}+\sum_{\bm{\tau}_{2}\ne\bm{0}}
		(C_{c}(\bm{\tau})+C_{d}(\bm{\tau}))
		\bm{z}_{2}^{\bm{\tau}_{2}}.
\end{equation}
Since $\bm{z}_{1}\cap\bm{z}_{2}=\varnothing$, comparing the two sides of Equation \eqref{eq: (A+B)(C+D)=2^m}, we have
\begin{equation}\label{eq: AA+BB=2^m}
	{A}(\bm{z}_{1})\cdot\overline{A}(\bm{z}_{1}^{-1})+{B}(\bm{z}_{1})\cdot\overline{B}(\bm{z}_{1}^{-1})=2^{m_{1}+1},
\end{equation}
\begin{equation}\label{eq: CC+DD=2^m}
	{C}(\bm{z}_{2})\cdot\overline{C}(\bm{z}_{2}^{-1})+{D}(\bm{z}_{2})\cdot\overline{D}(\bm{z}_{2}^{-1})=2^{m-m_{1}+1},
\end{equation}
which indicates 
$({a}(\bm{x}_{1}),{b}(\bm{x}_{1}))$ and $({c}(\bm{x}_{2}),{d}(\bm{x}_{2}))$ are both Golay array pairs.
\hfill\ensuremath{\square}
\end{proof}

\begin{proposition}\label{pro2}
The Golay array pair $({f}(\bm{x},{x}_{m+1}),{g}(\bm{x},{x}_{m+1}))$ in \eqref{eq: f=f0+f1} and \eqref{eq: g=g0+g1} are standard.
\end{proposition}

\begin{proof}
According to \eqref{factor-1}, \eqref{factor-2}, and Lemma \ref{lem: F=AB}, we have
\begin{equation}\label{eq: f_0=a+c}
	{f}_{0}(\bm{x})={a}(\bm{x}_{1})+{c}(\bm{x}_{2}),\qquad {g}_{0}(\bm{x})={b}(\bm{x}_{1})+{c}(\bm{x}_{2}).
\end{equation}
According to \eqref{eq: D(z)}, Lemma \ref{lem: F=AB} and Property \ref{property: F*(z)}, we have
\begin{equation}\label{eq: f_1=-b*+d}
	{g}_{1}(\bm{x})=-{a}^{*}(\bm{x}_{1})+\frac{q}{2}+{d}(\bm{x}_{2}),\qquad {f}_{1}(\bm{x})=-{b}^{*}(\bm{x}_{1})+{d}(\bm{x}_{2}).
\end{equation}

According to Proposition \ref{proposition1},
$({a}(\bm{x}_{1}),{b}(\bm{x}_{1}))$ and $({c}(\bm{x}_{2}),{d}(\bm{x}_{2}))$ form GAPs of size $\bm{2}^{m_1}$ and $\bm{2}^{m-m_1}$ respectively, where $0\leq{m}_1\leq{m}$.
Since Theorem \ref{thm: Type I} holds for ${n}\leq{m}$,
without loss of generality, suppose that 
\begin{equation}\label{eq: fun_a,b}
	\left\{\begin{aligned}
		{a}(\bm{x}_{1})&=\frac{q}{2}\sum_{k=1}^{m_{1}-1}x_{\pi(k)}x_{\pi(k+1)}+ \sum_{k=1}^{m_{1}}c_{k}x_{\pi(k)}+c_{0},\\
		{b}(\bm{x}_{1})&={a}(\bm{x}_{1})+\frac{q}{2}x_{\pi(1)}+{e},
	\end{aligned}\right.
\end{equation}
and 
\begin{equation}\label{eq: fun_c,d}
	\left\{\begin{aligned}
	{c}(\bm{x}_{2})&=\frac{q}{2}\sum_{k=m_1+2}^{m}x_{\pi(k)}x_{\pi(k+1)}+ \sum_{k=m_1+2}^{m+1}c_{k}x_{\pi(k)}+c'_{0},\\
	{d}(\bm{x}_{2})&={c}(\bm{x}_{2})+\frac{q}{2}x_{\pi(m_1+2)}+{e'},
	\end{aligned}\right.
\end{equation}
where $c_{k}$ ($0\leq{k}\leq{m+1}, k\ne{m_1+1}$), $c_{0}'$, $e, e'\in \Z_{q}$,
$\pi$ is a permutation of $\{1,2,\dots,m{+}1\}$ which satisfy
\begin{equation}
\left\{\begin{aligned}
&\{1,2,\dots,m_1\}\xrightarrow{\pi}\{1,2,\dots,m_1\},\\
&\{m_1+2,m_1+3,\dots,m+1\}\xrightarrow{\pi}\{m_1+1,m_1+2,\dots,m\},\\
&\pi(m_1{+}1)=m{+}1.
\end{aligned}\right.
\end{equation}
The reverse arrays of ${a}(\bm{x}_{1})$ and ${b}(\bm{x}_{1})$ are given by
\begin{equation}\label{eq: fun_a*,b*}
	\left\{\begin{aligned}
		{a}^{*}(\bm{x}_{1})&
		=\sum_{k=1}^{m_{1}-1}\frac{q}{2}x_{\pi(k)}x_{\pi(k+1)}
		-\sum_{k=1}^{m_{1}}c_{k}x_{\pi(k)}
		+\frac{q}{2}(x_{\pi(1)}+x_{\pi(m_1)})
		+\sum_{k=0}^{m_{1}}c_{k}
		+\frac{q}{2}(m_1-1),\\
		{b}^{*}(\bm{x}_{1})
		&={a}^{*}(\bm{x}_{1})+\frac{q}{2}x_{\pi(1)}+{e}+\frac{q}{2}.
	\end{aligned}\right.
\end{equation}

By substituting \eqref{eq: fun_a,b}, \eqref{eq: fun_c,d}, \eqref{eq: fun_a*,b*} into \eqref{eq: f_0=a+c}, \eqref{eq: f_1=-b*+d},
and substituting \eqref{eq: f_0=a+c}, \eqref{eq: f_1=-b*+d} into \eqref{eq: f=f0+f1}, \eqref{eq: g=g0+g1}, we have
\begin{equation}
\begin{split}
{f}(\bm{x},{x}_{m+1})
&=({a}(\bm{x}_{1})+{c}(\bm{x}_{2}))(1-{x}_{\pi(m_1+1)})+(-{b}^{*}(\bm{x}_{1})+{d}(\bm{x}_{2}))\cdot{x}_{\pi(m_1+1)}\\
&={a}(\bm{x}_{1})+{c}(\bm{x}_{2})+\left(-{a}(\bm{x}_{1})-{b}^{*}(\bm{x}_{1})-{c}(\bm{x}_{2})+{d}(\bm{x}_{2})\right)\cdot{x}_{\pi(m_1+1)}\\
&={a}(\bm{x}_{1})+{c}(\bm{x}_{2})
+\frac{q}{2}{x}_{\pi(m_1)}\cdot{x}_{\pi(m_1+1)}+\frac{q}{2}x_{\pi(m_1+2)}\cdot{x}_{\pi(m_1+1)}\\
&\quad
+\left(\frac{q}{2}m_1-\sum_{k=0}^{m_{1}}c_{k}-c_{0}+{e'}-{e}\right)\cdot{x}_{\pi(m_1+1)}\\
&=\sum_{k=1}^{m}\frac{q}{2}x_{\pi(k)}x_{\pi(k+1)}+ \sum_{k=1}^{m+1}c_{k}x_{\pi(k)}+c_{0}+c'_{0},
\end{split}
\end{equation}
where $c_{m_1+1}=\frac{q}{2}m_1-\sum_{k=0}^{m_{1}}c_{k}-c_{0}+{e'}-{e}$
and
\begin{equation}
\begin{split}
{g}(\bm{x},{x}_{m+1})
&=({b}(\bm{x}_{1})+{c}(\bm{x}_{2}))(1-{x}_{\pi(m_1+1)})+\left(-{a}^{*}(\bm{x}_{1})+\frac{q}{2}+{d}(\bm{x}_{2})\right)\cdot{x}_{\pi(m_1+1)}\\
&={b}(\bm{x}_{1})+{c}(\bm{x}_{2})
+\left(-{b}(\bm{x}_{1})-{a}^{*}(\bm{x}_{1})-{c}(\bm{x}_{2})+{d}(\bm{x}_{2})+\frac{q}{2} \right)\cdot{x}_{\pi(m_1+1)}\\
&={f}(\bm{x},{x}_{\pi(m_1+1)})+\frac{q}{2}x_{\pi(1)}+{e}.\\
\end{split}
\end{equation}
Consequently, we have the Golay array pair $({f}(\bm{x},{x}_{m+1}),{g}(\bm{x},{x}_{m+1}))$ must be standard.
\hfill\ensuremath{\square}
\end{proof}

Proposition \ref{pro2} implies Theorem \ref{thm: Type I} holds for array size $\bm{2}^{(m+1)}$, which complete the process of mathematical induction.

\section*{Declarations}
\begin{itemize}
\item \textbf{Conflict of interest} 
The authors declare that they have no conflicts of interest relevant to the content of this article.
\item \textbf{Availability of data and materials} 
Not applicable.
\end{itemize}


\begin{thebibliography}{00}
	
	
	
	
	\bibitem{Chai2021DCCWalsh}
	J. Chai, Z. Wang, and E. Xue,
	\newblock \lq\lq Walsh Spectrum and Nega Spectrum of Complementary Arrays,\rq\rq \
	{\em Designs Codes and Cryptography},
	vol. 89, pp. 2663--2677, 2021.	
	
	
	\bibitem{Davis1999Peak}
	J. A. Davis and J. Jedwab,
	\newblock \lq\lq Peak-to-mean power control in OFDM, Golay complementary sequences, and Reed-Muller codes,\rq\rq \
	{\em IEEE Trans. Inf. Theory},
	vol. 45, no. 7, pp. 2397--2417, 1999.	
	
	\bibitem{Algebra}
	D. S. Dummit and R. M. Foote,
	\newblock  {\em Abstract Algebra (3rd Edition)},
	John Wiley and Sons, 2004.	
	
	\bibitem{Fiedler2008Am}
	F. Fiedler, J. Jedwab, and M. G. Parker,
	\newblock \lq\lq A multi-dimensional approach to the construction and enumeration of Golay complementary sequences,\rq\rq
	{\em Journal of Combinatorial Theory, Series A},
	vol. 115, no. 5, pp. 753--776, 2008.	
	
	
	\bibitem{Golay1951Static}
	M. J. Golay,
	\newblock \lq\lq static multislit spectrometry and its application to the panoramic display of infrared spectra,\rq\rq
	{\em Journal of the Optical Society of America},
	vol. 47, no. 7, pp. 468--472, 1951.	
	
	
	\bibitem{Jedwab2007Golay}
	J. Jedwab and M. G. Parker,
	\newblock \lq\lq Golay complementary array pairs,\rq\rq \
	{\em Designs Codes and Cryptography},
	vol. 44, no. 7, pp. 209--216, 2007.	
	
	
	
	
	
	
	
	
%
	
%
	
	
\bibitem{Li2005More}
Y.~Li and W.~B.~Chu,
\newblock \lq\lq More Golay sequences,\rq\rq
\newblock {\em  IEEE Trans. Inf. Theory}, vol. 51, no. 3, pp. 1141--1145, 2005.
%
	
\bibitem{Fiedler2006How}
F.~Fiedler and J.~Jedwab,
\newblock \lq\lq How do more Golay sequences arise?\rq\rq
\newblock {\em  IEEE Trans. Inf. Theory}, vol. 52, no. 9, pp. 4261--4266, 2006.

\bibitem{Fiedler2008A}
F.~Fiedler, J.~Jedwab and  M.~G.~Parker,
\newblock \lq\lq A framework for the construction of Golay sequences,\rq\rq
\newblock {\em  IEEE Trans. Inf. Theory}, vol. 54, no. 7, pp.  3114--3129, 2008.

\bibitem{Borwein2004A}
P. B.~Borwein, R. A.~Ferguson,
\newblock \lq\lq A complete description of Golay pairs for lengths up to $100$,\rq\rq
\newblock {\em Mathmatics of Computation}, vol. 73, no. 246, pp. 967--985, 2004.

\bibitem{Dymond1992}
M. Dymond, Barker, {\em Arrays: Existence, Generalization and Alternatives}, PhD thesis, University of London, 1992.

\bibitem{Luke1985Welti}
H. D, L{\"u}ke,
\newblock \lq\lq Sets of one and higher dimensional Welti codes and complementary codes,\rq\rq
\newblock{\em Trans. Aerospace Electron. Syst.},
vol. 21, pp. 170--179, 1985.

\bibitem{Paterson00}
K.~G.~Paterson, \lq\lq Generalized Reed-Muller codes and power control in
OFDM modulation,\rq\rq \ {\em IEEE Trans. Inf. Theory}, vol. 46, no.
1, pp. 104--120, 2000.



\bibitem{Turyn1974}
R. Turyn,
\newblock \lq\lq Hadamard matrices, Baumert-Hall units, four-symbol sequences, pulse compression, and surface wave encodings\rq\rq
{\em Journal of Combinatorial Theory, Series A},
vol. 16,  pp. 313--333, 1974.	
	
\end{thebibliography}
\end{document}